\documentclass[11pt]{article}

\usepackage[utf8]{inputenc}
\usepackage[T1]{fontenc}
\usepackage{lmodern}
\usepackage[protrusion=true,expansion=false]{microtype}
\usepackage[margin=1in, top=1.1in, bottom=1.1in]{geometry}

\usepackage{amsmath,amssymb,amsthm}
\usepackage{booktabs}
\usepackage{graphicx}
\usepackage[colorlinks=true,
            linkcolor=blue!70!black,
            citecolor=blue!70!black,
            urlcolor=blue!70!black]{hyperref}
\usepackage{xcolor}
\usepackage{enumitem}
\usepackage[numbers,sort&compress]{natbib}
\usepackage{caption}
\usepackage{subcaption}
\usepackage{tikz}
\usetikzlibrary{shapes.geometric,arrows.meta,positioning,fit,backgrounds,calc}
\usepackage{listings}
\usepackage{float}
\usepackage{array}
\usepackage{multirow}

\usepackage{titlesec}
\titleformat{\section}{\large\bfseries}{\thesection}{1em}{}
\titleformat{\subsection}{\normalsize\bfseries}{\thesubsection}{1em}{}

\newtheorem{theorem}{Theorem}[section]

\newtheorem{corollary}[theorem]{Corollary}
\newtheorem{proposition}[theorem]{Proposition}
\newtheorem{definition}{Definition}[section]
\newtheorem{assumption}{Assumption}[section]

\newtheorem{remark}{Remark}[section]

\newcommand{\Sr}{S_r}
\newcommand{\Sp}{S_p}
\newcommand{\Ar}{A_r}
\newcommand{\Ap}{A_p}
\newcommand{\St}[1]{S_r(#1)}

\newcommand{\Ndefval}{\bot}
\newcommand{\IER}{\mathrm{IER}}
\newcommand{\SHR}{\mathrm{SHR}}
\newcommand{\OCR}{\mathrm{OCR}}

\title{\textbf{Reconstructive Authority Model:}\\
       \textbf{Runtime Execution Validity Under Partial Observability}\\[0.5em]
       \large\normalfont Agent Governance Series --- Paper~5}

\author{Marcelo Fernandez\\
        \small\textit{Independent Research}\\
        \small\texttt{DOI: 10.5281/zenodo.19669430}}

\date{April 2026}

\begin{document}
\maketitle

\begin{center}
\small
\begin{tabular}{clll}
\toprule
\textbf{Paper} & \textbf{Short title} & \textbf{Zenodo DOI} & \textbf{arXiv}\\
\midrule
P0 & Atomic Decision Boundaries~\cite{fernandez2026a}     & 10.5281/zenodo.19670649 & arXiv:2604.17511\\
P1 & Agent Control Protocol (ACP)~\cite{fernandez2026b}   & 10.5281/zenodo.19672575 & arXiv:2603.18829\\
P2 & From Admission to Invariants (IML)~\cite{fernandez2026c} & 10.5281/zenodo.19672589 & arXiv:2604.17517\\
P3/4 & Irreducible Governance Structure~\cite{fernandez2026govstr} & 10.5281/zenodo.19708496 & TBD\\
P5 & \textbf{Reconstructive Authority Model (RAM)} & 10.5281/zenodo.19669430 & TBD\\
P6 & Operationalizing Reconstructive Authority~\cite{fernandez2026op} & 10.5281/zenodo.19699460 & TBD\\
\bottomrule
\end{tabular}
\end{center}

\begin{abstract}
Autonomous systems increasingly operate under partial observability where the
true execution-relevant state $\Sr(t)$ is never fully accessible.
Existing governance mechanisms---trusted execution environments, oracle-signed
state proofs, cryptographic attestation---enforce the integrity of computation
and state projections.
We show this is structurally insufficient: \emph{an authenticated projection
of state is necessary, never sufficient}.

We separate integrity from coverage: authenticated projection is necessary,
never sufficient.
The \emph{Reconstructive Authority Model} (RAM) introduces a reconstruction
gate that reasons over an explicit \emph{coverage envelope}---comprising the
proven state, declared assumptions, and the unobservable residual whose
existence is established by Paper~2 (IML,~\cite{fernandez2026c})---and permits
execution only when coverage is adequate for that action class; otherwise it
narrows privileges dynamically or fails closed.
Attestation proves trust in measurement, not completeness of
execution-relevant reality.

We formalize RAM, prove its necessity via two theorems and three corollaries,
present a hybrid RAM\,+\,Attestation architecture with privilege-narrowing,
and demonstrate through case study and synthetic experiments that RAM eliminates
invalid execution under all coverage levels while attestation-based systems
exhibit invalid execution rates proportional to $(1 - |\Sp|/|\Sr|)$.

This paper is Paper~5 of the Agent Governance Series.
It operationalizes the observability impossibility of
Paper~2~(IML,~\cite{fernandez2026c}) and completes the execution semantics of
the four-layer architecture of Paper~3/4~(\cite{fernandez2026govstr}).
\end{abstract}

\tableofcontents
\newpage

\section{Introduction}
\label{sec:intro}

Autonomous agent systems face a fundamental tension: decisions must be made at
admission time based on available state, but execution happens later, when
conditions may have changed.
Existing governance frameworks address this through \emph{attestation}---cryptographic
mechanisms that verify consistency between the current execution environment and
a previously approved state.
Trusted Execution Environments (TEEs), oracle-signed state proofs, and
justification-based control systems all operate on this principle.

The limitation of attestation-based governance is structural, not implementational.
We \textbf{separate integrity from coverage}: attestation resolves the former
but is silent on the latter.
Integrity---that a computation was performed correctly on the observed state---
can be cryptographically guaranteed.
Coverage---that the observed state was \emph{sufficient} to justify the action
in the first place---cannot.

An authenticated projection of state is necessary, never sufficient.
Since the provable state $\Sp \subsetneq \Sr$ strictly in any real system,
there always exists a component $\delta = \Sr \setminus \Sp$ that influences
execution validity but lies outside the attested model.
When $\delta$ becomes execution-critical---due to runtime drift, emergent
conditions, or delayed observability---attestation produces a false positive:
a system that \emph{appears} consistent but \emph{is not valid} to execute.

We formalize this limitation and introduce the \textbf{Reconstructive Authority
Model (RAM)}: a governance framework that shifts the fundamental question from
\emph{``Is execution still consistent with what was approved?''} to
\emph{``Can authority still be constructed from what is true now?''}
RAM reasons over a \emph{coverage envelope} and permits action only when
coverage is adequate for the specific action class---narrowing privileges when
coverage is partial, failing closed when it is insufficient.

\paragraph{Contributions.}
\begin{enumerate}[leftmargin=*, label=(\roman*)]
  \item A formal characterization of the structural limitation of attestation-based
        governance systems, with proof that attestation correctness does not imply
        execution validity (Section~\ref{sec:limits}).
  \item The RAM formal model: definitions, non-persistence theorem, invariant
        treatment, and partial observability handling (Section~\ref{sec:ram}).
  \item Two theorems establishing (a)~the limits of attestation without
        reconstruction and (b)~the necessity of reconstruction for execution
        validity (Section~\ref{sec:theory}).
  \item A RAM+Attestation hybrid architecture with formal failure handling
        (Section~\ref{sec:arch}).
  \item A case study and synthetic experimental evaluation comparing
        Attestation, Oracle, and RAM models under drift injection
        (Sections~\ref{sec:case}--\ref{sec:exp}).
  \item An account of RAM's position within the Agent Governance Series
        (Section~\ref{sec:series}).
\end{enumerate}

\paragraph{Series context.}
The Agent Governance Series develops a complete formal theory of autonomous
agent governance.
Paper~0~\cite{fernandez2026a} establishes atomic decision boundaries---the
structural requirement for guaranteeing execution-time admissibility.
Paper~1 (ACP)~\cite{fernandez2026b} implements enforcement via admission control
at the policy boundary.
Paper~2 (IML)~\cite{fernandez2026c} proves the epistemological limits of local
observability: an agent can never fully observe the state space relevant to its
own invariants.
Paper~3/4~\cite{fernandez2026govstr} addresses both distributive limits
under strategy-proof allocation mechanisms and proves the irreducibility of the
four-layer governance architecture.
This paper (P5) closes the series by addressing the runtime question: given
that observability is incomplete (P2) and the architecture is irreducible (P3/4),
how should a system decide whether to execute at all?

\section{Background}
\label{sec:background}

\subsection{Attestation-Based Governance}

Trusted Execution Environments (TEEs) such as Intel SGX~\cite{costan2016intel}
provide hardware-backed attestation: a guarantee that a computation runs on
unmodified code within a tamper-resistant enclave.
At the governance level, TEE-based systems combine attestation with
justification-based control~\cite{gu2017certikos}, where an agent generates a
justification block $J$ prior to execution, and a TEE verifies that execution
proceeds only if $J$ is consistent with attested state.

Oracle-extended systems improve coverage by augmenting the attested state $\Sp$
with externally-signed state proofs $\Sp^{\text{oracle}}$, reducing the gap
between $\Sp$ and $\Sr$~\cite{adler2018astraea}.
Architectures such as ATLAS~\cite{zhang2020deco} combine TEE attestation with
oracle-anchored state proofs to increase the trustworthy observable surface.
However, as we prove formally (Corollary~\ref{cor:oracle}), no finite proof
system can achieve $\Sp = \Sr$ in a general system with partial observability
and dynamic state: oracle extensions reduce but cannot eliminate the state gap.

Runtime verification (RV) approaches~\cite{leucker2009brief, falcone2012runtime,
havelund2001java} monitor execution against formal specifications, typically
through monitors that observe traces and raise alarms on specification violations.
A critical distinction: RV monitors assume the monitored state is complete---that
what they observe is sufficient to evaluate the specification.
RAM relaxes this assumption explicitly.
Where an RV monitor asks \emph{``does the observed trace satisfy the property?''},
RAM asks \emph{``is what we observe sufficient to determine whether execution
should proceed at all?''}
This makes RAM complementary to RV: RV governs \emph{how} execution proceeds;
RAM governs \emph{whether} it may.

\subsection{Partial Observability}

Partial observability is a fundamental feature of real-world agent systems,
studied extensively in the context of Partially Observable Markov Decision
Processes (POMDPs)~\cite{kaelbling1998planning}.
The POMDP framework models agents that receive observations rather than direct
state access, maintaining belief distributions over the real state.
RAM shares this epistemic humility but addresses a different problem: not
optimal policy under uncertainty, but binary authority---whether the system is
\emph{permitted to act at all}.

\subsection{Agent Governance Series}

The governance series establishes a formal framework for agent systems under
finite observability.
Paper~2 (IML)~\cite{fernandez2026c} is most directly relevant here:
it proves formally that for any governance architecture with finite memory and
local observation, there exists a violation path that is undetectable until it
has already caused a constraint breach.
RAM operationalizes the response to this result: since full observability is
unachievable (IML), authority must be constructed conservatively from whatever
is observable, with explicit failure when observation is insufficient.

\section{Limitations of Attestation-Based Governance}
\label{sec:limits}

We identify six structural limitations of attestation-based governance.
These are not implementation failures; they are consequences of the gap
between $\Sp$ and $\Sr$.

\subsection{Integrity Does Not Imply Semantic Correctness}

Attestation mechanisms guarantee that a computation has not been tampered with
and executes within a trusted boundary.
They do not guarantee that the \emph{content} of that computation is semantically
correct or complete relative to the execution context.

When execution authority is derived from internally generated artifacts (justification
blocks), the system is vulnerable to:
\begin{itemize}
  \item Incomplete representations of state;
  \item Omitted constraints;
  \item Structurally valid but semantically insufficient reconstructions.
\end{itemize}

This leads to the first failure class:
a justification is valid, attestable, and internally consistent, yet insufficient
to support safe execution.

\subsection{The Self-Justification Boundary}

Even when attestation enforces integrity, if the system under control is
responsible for generating its own justification artifacts, a structural
dependency emerges:

\begin{center}
\textit{Execution authority depends on the system's ability to correctly describe
the conditions under which it is allowed to act.}
\end{center}

This creates a failure mode where:
(1)~the system produces a formally valid justification;
(2)~the attestation verifies its integrity;
(3)~execution proceeds;
yet the justification does not fully capture the relevant execution constraints.
This is not a failure of \emph{enforcement}, but of \emph{epistemic completeness}.

\subsection{External State Proofs and the Coverage Problem}

To address self-justification, some architectures introduce externally signed
state proofs (oracles, ledgers).
While this improves trust, it introduces a new limitation:
the system operates against a \emph{provable projection} of reality, not reality
itself.

No external proof system can guarantee complete coverage of all execution-relevant
state, due to:
\begin{itemize}
  \item \textbf{Observability limits}: not all variables are measurable;
  \item \textbf{Latency}: state may change between proof generation and execution;
  \item \textbf{Model constraints}: only predefined dimensions are captured.
\end{itemize}

Thus, even perfect comparison between internal intent and external proof can
result in correct validation against an incomplete state model.

\subsection{Drift Outside the Proof Boundary}

Attestation-based systems typically detect divergence as mismatches between:
(i)~attested state at admission and (ii)~current state at execution.
However, divergence can also arise from state dimensions that were
\emph{never included} in the attested or provable model.

This creates a blind spot: no mismatch is detected, the system remains
``within envelope'', execution continues---despite the fact that the true
execution basis has shifted.

\subsection{Enforcement Without Reconstruction}

Attestation-based governance models are fundamentally \emph{enforcement-oriented}:
they verify consistency, enforce constraints, and terminate on mismatch.
However, they do not inherently address whether execution authority can still
be \emph{constructed} from the current state.

A system may pass all attestation checks, remain within defined envelopes, and
maintain cryptographic validity---and still lack a valid basis for execution
under actual conditions.

\subsection{Formal Limitation}
\label{subsec:formal-limit}

\begin{definition}[State gap]
Let $\Sr(t)$ be the real execution-relevant state at time $t$,
$\Sp(t)$ the provable state accessible to attestation, and
$J(t)$ the justification derived from $\Sp(t)$.
The \emph{state gap} at time $t$ is
$\delta(t) = \Sr(t) \setminus \Sp(t)$.
\end{definition}

Under any attestation model, even when $J(t)$ is valid relative to $\Sp(t)$,
$\Sp(t) \subseteq \Sr(t)$, and all attestation checks pass, it does not follow
that $J(t)$ is valid relative to $\Sr(t)$.

Therefore: \emph{attestation correctness does not imply execution correctness}.

\section{Reconstructive Authority Model}
\label{sec:ram}

\subsection{Core Principle}

\begin{center}
\textit{Execution authority is not a persistent property granted at admission,
but a derivable property that must be continuously reconstructible from the
current state.}
\end{center}

Execution is valid at time $t$ if and only if authority can be constructed
from $\Sr(t)$.
This removes dependency on historical decisions, stored justifications, and
previously attested envelopes.

\subsection{Positioning Against Related Approaches}
\label{subsec:positioning}

RAM must be distinguished from two related but distinct concepts.

\textbf{Conservative halt policies} define a rule: ``halt if uncertainty
exceeds threshold $\theta$.''
This requires a probability measure over states and a tunable threshold.
RAM does not require either: $F$ is not a probabilistic function but a
deterministic constructor that returns $\Ndefval$ when required state
components are missing---not when uncertainty is high.
RAM can halt with certainty (when $F$ is definitively $\text{false}$) and can
execute under uncertainty (when available components suffice to determine
$F = \text{true}$).

\textbf{POMDP-based abstention} requires a belief state $b(s)$ over $\Sr$
and a value function $V(b)$ over beliefs.
RAM requires neither: it operates on the observable state $\hat{S}_r(t)$
directly, returning $\Ndefval$ when authority is not constructible from what
is available---a structural condition, not a probabilistic one.

The contribution of RAM is not ``halt when uncertain'' but:
\emph{define execution authority as a constructible property of the current
state, such that execution is permitted only when authority can be derived}.
This is a different problem from uncertainty quantification or belief-based
decision-making.

\subsection{Coverage Envelope}
\label{subsec:envelope}

\begin{definition}[Coverage Envelope]
The \emph{coverage envelope} $\mathcal{E}(t)$ at time $t$ is the triple:
\[
\mathcal{E}(t) \;=\; \bigl(\Sp(t),\;\mathcal{H}(t),\;\delta(t)\bigr)
\]
where:
\begin{itemize}
  \item $\Sp(t) \subseteq \Sr(t)$: the \emph{proven state}---components that
        are observable and cryptographically attestable at time $t$;
  \item $\mathcal{H}(t)$: \emph{declared assumptions}---explicit propositions
        about unobserved state that the system accepts as operationally
        necessary for authority construction;
  \item $\delta(t) = \Sr(t) \setminus \Sp(t)$: the \emph{acknowledged
        unobservable residual}---the gap whose permanent existence is
        established by Paper~2 (IML,~\cite{fernandez2026c}).
\end{itemize}
Coverage is \emph{adequate for action class $\alpha$} if the authority
constructor can derive $A(t) = \text{true}$ from $\mathcal{E}(t)$ under the
requirements of $\alpha$.
Coverage is \emph{partial for $\alpha$} if it suffices for a strict subset
$\alpha' \subsetneq \alpha$.
Coverage is \emph{insufficient} if no $\alpha' \subseteq \alpha$ can be
authorized.
\end{definition}

\begin{remark}
Making $\delta(t)$ explicit in the envelope is not merely bookkeeping: it
forces system designers to acknowledge \emph{what they do not know} at each
execution step, rather than silently assuming the provable state is complete.
This is the operational consequence of IML's impossibility result.
\end{remark}

\subsection{Formal Definition}

\begin{definition}[Reconstructive Authority]
Let $\Sr(t)$ be the real execution-relevant state at time $t$,
$\mathcal{A}$ the space of action classes, $A(t)$ the execution authority
at time $t$, and
$F: \mathcal{E} \times \mathcal{A} \to \mathcal{A} \cup \{\text{false}, \Ndefval\}$
the authority construction function over coverage envelopes.
Then:
\[
A(t) = F\!\left(\mathcal{E}(t),\, \alpha\right)
\]
where $\alpha$ is the requested action class.
\begin{itemize}
  \item If $F(\mathcal{E}(t), \alpha) = \alpha$: execute at full scope;
  \item If $F(\mathcal{E}(t), \alpha) = \alpha' \subsetneq \alpha$:
        execute with narrowed privileges ($\alpha'$ only);
  \item If $F(\mathcal{E}(t), \alpha) \in \{\text{false}, \Ndefval\}$:
        halt or defer.
\end{itemize}
The value $\Ndefval$ represents \emph{insufficient information} to determine
authority for any sub-class; $\text{false}$ represents a definitive refusal.
\end{definition}

The inclusion of $\Ndefval$ is critical: when the observable state is
insufficient to establish authority, RAM mandates halting rather than
defaulting to execution.

\subsection{Non-Persistence of Authority}

Unlike admission-based models, authority in RAM is non-persistent:

\begin{proposition}[Non-persistence]
\[
A(t_0) = \text{true} \;\not\Rightarrow\; A(t_1) = \text{true} \quad (t_1 > t_0)
\]
even when no explicit policy violation is detected between $t_0$ and $t_1$.
\end{proposition}

Authority must be recomputed at each execution step, not assumed from prior
successful computations.
This is the key operational difference from attestation-based models.

\subsection{Invariants as Constructive Conditions}

In RAM, invariants are not validation checks but \emph{constructive requirements}.

Instead of verifying ``invariant holds,'' we define:
$F(\St{t})$ is defined only if invariants hold over $\St{t}$.
If invariants fail, $F$ becomes undefined ($\Ndefval$), authority cannot be constructed,
and execution must halt.

This shifts the semantics from \emph{reactive} (detect violation, then stop)
to \emph{constructive} (establish validity, then proceed).

\subsection{Partial Observability Handling}

RAM explicitly acknowledges that $\Sr(t)$ is never fully observable.
Authority construction therefore operates under bounded uncertainty with
explicit incompleteness:

\[
F(\St{t}) = \Ndefval \;\text{ if required components of } \Sr(t) \text{ are unknown.}
\]

The default behavior is \textbf{halt or defer}, not continue under partial
justification.
This is the conservative assumption established in Paper~2 (IML)~\cite{fernandez2026c}
as the only safe default under finite observability.

\subsection{Drift Interpretation}

In attestation-based models, drift is defined as a detectable mismatch between
states at different points in time.
Under RAM, drift is not defined as a mismatch but as \emph{inability to
reconstruct authority from current conditions}:

\[
\text{Drift at } t \;\equiv\; F(\St{t}) \in \{\text{false}, \Ndefval\}
\]

This eliminates the need to define, detect, or classify drift explicitly.
Instead, drift manifests automatically as an authority failure.

\subsection{Execution Model}

At each execution step, the RAM execution loop is:

\begin{enumerate}[leftmargin=*, label=(\arabic*)]
  \item Observe current state $\hat{S}_r(t) \subseteq \Sr(t)$ (partial, bounded);
  \item Attempt to compute $A(t) = F(\hat{S}_r(t))$;
  \item If $A(t) = \text{true}$: proceed with execution;
  \item If $A(t) \in \{\text{false}, \Ndefval\}$: halt or re-resolve.
\end{enumerate}

No execution step occurs without a fresh authority computation.
Historical authority grants are not consulted.

\section{Theoretical Results}
\label{sec:theory}

\subsection{Theorem 1: Limits of Attestation Without Reconstruction}

\begin{theorem}[Attestation insufficiency]
\label{thm:attestation-insufficiency}
No attestation-based governance system can guarantee execution validity with
respect to the real execution state $\Sr$ unless execution authority is
reconstructible from $\Sr(t)$ at runtime.
\end{theorem}

\paragraph{Formal setup.}
Let $\Sp(t) \subseteq \Sr(t)$ be the provable state,
$J(t)$ the justification derived from $\Sp(t)$,
$\Ap(t) = G(\Sp(t))$ the attestation-based authority decision for some
deterministic function $G$, and $\Ar(t) = F(\Sr(t))$ the real authority
decision.
Execution validity requires $\Ar(t) = \text{true}$.

\paragraph{Assumptions.}
\begin{assumption}[Partial observability]
$\Sp(t) \subsetneq \Sr(t)$ in general (strict inclusion).
\end{assumption}
\begin{assumption}[Attestation correctness]
All attestation mechanisms correctly verify the integrity of $J(t)$ and the
authenticity of $\Sp(t)$.
\end{assumption}
\begin{assumption}[Deterministic evaluation]
$\Ap(t) = G(\Sp(t))$ for some deterministic decision function $G$.
\end{assumption}

\begin{proof}
Since $\Sp(t) \subsetneq \Sr(t)$, there exists at least one component
$\delta(t) \in \Sr(t) \setminus \Sp(t)$.
This component represents execution-relevant state not captured in the
attested model.

Construct the following case:
\begin{enumerate}[leftmargin=*, label=(\roman*)]
  \item $G(\Sp(t)) = \text{true}$ (no violation detectable in provable state);
  \item $\delta(t)$ invalidates execution under real conditions, i.e.,
        $F(\Sr(t)) = \text{false}$.
\end{enumerate}

Such a $\delta(t)$ exists by construction whenever the state space admits
conditions that are execution-relevant but not covered by $\Sp(t)$.

Under these conditions: $\Ap(t) = \text{true}$ and $\Ar(t) = \text{false}$.

Since attestation operates only over $\Sp(t)$, it cannot detect $\delta(t)$.
Therefore, execution proceeds under attestation while being invalid under
real state.

Attestation correctness over $\Sp(t)$ does not imply execution validity
over $\Sr(t)$.
\end{proof}

\begin{corollary}[Enforcement limitation]
\label{cor:enforcement}
Even with perfect enforcement (e.g., TEE termination on detected mismatch):
if no mismatch is observable in $\Sp(t)$, no enforcement is triggered.
Therefore, enforcement cannot prevent invalid execution when the invalidating
condition lies outside the provable state.
\end{corollary}

\begin{corollary}[Oracle limitation]
\label{cor:oracle}
Extending $\Sp(t)$ via external proofs (oracles, ledgers) reduces but does
not eliminate the state gap:
\[
\Sp^{\text{oracle}}(t) = \Sp(t) \cup \Sp^{\text{ext}}(t) \subsetneq \Sr(t)
\]
Therefore, no finite proof system can guarantee full coverage of $\Sr(t)$.
\end{corollary}

\begin{corollary}[Justification limitation]
\label{cor:justification}
Even when $J(t)$ is complete relative to $\Sp(t)$ and correctly attested,
it may be incomplete relative to $\Sr(t)$, yielding a structurally valid
justification for an invalid execution.
\end{corollary}

\subsection{Theorem 2: Necessity of Reconstruction}

\begin{theorem}[Necessity of RAM]
\label{thm:necessity}
Execution validity can only be guaranteed if execution authority is defined
as a function over $\Sr(t)$:
\[
A(t) = F\!\left(\St{t}\right)
\]
and execution is permitted if and only if $F(\St{t}) = \text{true}$.
Any system that relies solely on $\Sp(t)$ or on historical attestation
cannot guarantee $\Ar(t) = \text{true}$.
\end{theorem}

\begin{proof}[Proof sketch]
By Theorem~\ref{thm:attestation-insufficiency}, $\exists\, \Sr(t)$ such that
$\Ap(t) = \text{true}$ but $\Ar(t) = \text{false}$.
Therefore, attestation-based authority $\Ap$ is not equivalent to real
authority $\Ar$.

The only decision function that is guaranteed to agree with $\Ar(t)$ is one
that takes $\Sr(t)$ as input directly.
Since $\Sr(t)$ is not fully observable in general, such a function must
return $\Ndefval$ when the available observation $\hat{S}_r(t) \subsetneq \Sr(t)$
is insufficient to determine $\Ar(t)$.
The conservative default under $\Ndefval$ is halt, which avoids invalid execution
at the cost of potentially halting valid executions.
\end{proof}

\begin{remark}
Theorem~\ref{thm:necessity} establishes that the conservative halting behavior
of RAM is not a design choice but a logical consequence of the partial
observability constraint proved in Paper~2 (IML)~\cite{fernandez2026c}.
Systems that wish to avoid invalid execution under partial observability
\emph{must} be willing to halt under uncertainty.
\end{remark}

\section{Contrast: RAM vs Attestation-Based Models}
\label{sec:contrast}

\begin{table}[H]
\centering
\caption{RAM vs Attestation-Based Models: structural comparison}
\label{tab:comparison}
\begin{tabular}{lll}
\toprule
\textbf{Property} & \textbf{Attestation-Based} & \textbf{RAM} \\
\midrule
Authority source      & Derived from $\Sp(t)$              & Derived from $\Sr(t)$ \\
Role of history       & Strong (admission, envelopes)       & None (history irrelevant) \\
Drift handling        & Detect mismatch vs.\ attested state & Authority fails to reconstruct \\
Failure mode          & Detected divergence                 & Undefined authority ($\Ndefval$) \\
State basis           & $\Sp \subseteq \Sr$ (provable) & $\Sr(t)$ (actual, possibly partial) \\
Invariant role        & Validation checks                   & Constructive conditions \\
Default under uncertainty & Continue                       & Halt \\
Privilege adaptation  & Static (all-or-nothing)      & Dynamic narrowing ($\alpha \!\to\! \alpha'$) \\
Core limitation       & Incomplete state coverage           & Conservative halting \\
\bottomrule
\end{tabular}
\end{table}

\paragraph{Philosophical difference.}

\medskip
\begin{center}
\begin{tabular}{p{0.45\textwidth}p{0.45\textwidth}}
\textbf{Attestation-Based:} & \textbf{RAM:} \\
``Ensure execution remains consistent with what was approved.'' &
``Allow execution only if it can be justified \emph{now}.'' \\
\end{tabular}
\end{center}
\medskip

\paragraph{Bridging insight.}
These models are not mutually exclusive.
Attestation provides integrity, enforcement, and tamper resistance.
RAM provides semantic correctness, runtime validity, and authority grounding.
A robust system requires both: attestation to guarantee \emph{how} computation
executes, RAM to determine \emph{whether} it should execute at all.

\begin{center}
\textit{Attestation can guarantee that a system executes correctly.\\
RAM determines whether it should execute at all.}
\end{center}

\section{System Architecture}
\label{sec:arch}

\subsection{Component Overview}

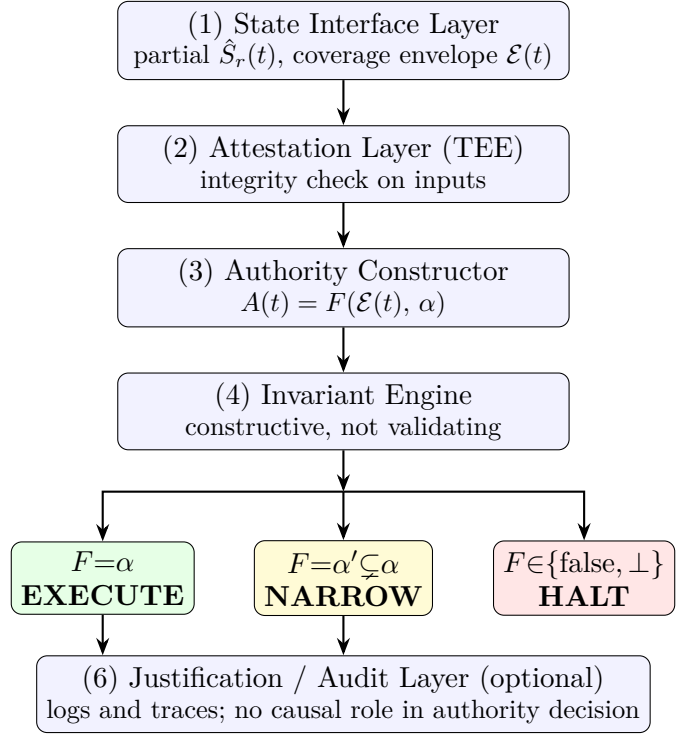
\begin{figure}[H]
\centering
\begin{tikzpicture}[
  box/.style={draw, rounded corners=4pt, minimum width=6cm, minimum height=0.9cm,
              align=center, fill=blue!5},
  gate/.style={draw, rounded corners=4pt, minimum width=2.3cm, minimum height=0.9cm,
               align=center, fill=green!10},
  narr/.style={draw, rounded corners=4pt, minimum width=2.3cm, minimum height=0.9cm,
               align=center, fill=yellow!20},
  halt/.style={draw, rounded corners=4pt, minimum width=2.3cm, minimum height=0.9cm,
               align=center, fill=red!10},
  arr/.style={-Stealth, thick},
  node distance=0.55cm
]
\node[box] (state)  {(1) State Interface Layer\\[-2pt]
                      {\small partial $\hat{S}_r(t)$, coverage envelope $\mathcal{E}(t)$}};
\node[box, below=0.6cm of state]  (attest) {(2) Attestation Layer (TEE)\\[-2pt]
                      {\small integrity check on inputs}};
\node[box, below=0.6cm of attest] (constr) {(3) Authority Constructor\\[-2pt]
                      {\small $A(t) = F(\mathcal{E}(t),\,\alpha)$}};
\node[box, below=0.6cm of constr] (invar)  {(4) Invariant Engine\\[-2pt]
                      {\small constructive, not validating}};

\coordinate (split) at ($(invar.south)+(0,-0.55)$);

\node[gate] (exec)   at ($(split)+(-3.2,-1.15)$)
  {$F{=}\alpha$\\[-1pt]\textbf{EXECUTE}};
\node[narr] (narrow) at ($(split)+(0,-1.15)$)
  {$F{=}\alpha'{\subsetneq}\alpha$\\[-1pt]\textbf{NARROW}};
\node[halt]  (halt)   at ($(split)+(+3.2,-1.15)$)
  {$F{\in}\{\text{false},\Ndefval\}$\\[-1pt]\textbf{HALT}};

\node[box, minimum width=6cm]
  (audit) at ($(exec.south)!0.5!(halt.south)+(0,-1.05)$)
  {(6) Justification / Audit Layer (optional)\\[-2pt]
   {\small logs and traces; no causal role in authority decision}};

\draw[arr] (state)  -- (attest);
\draw[arr] (attest) -- (constr);
\draw[arr] (constr) -- (invar);

\draw[arr] (invar.south) -- (split);
\draw[arr] (split) -| (exec.north);
\draw[arr] (split)  -- (narrow.north);
\draw[arr] (split) -| (halt.north);

\draw[arr] (exec.south) -- (exec.south |- audit.north);
\draw[arr] (narrow.south) -- (audit.north);
\end{tikzpicture}
\caption{RAM reconstruction gate with three outcomes.
The authority constructor $F(\mathcal{E}(t),\alpha)$ evaluates the coverage
envelope against the requested action class $\alpha$:
full coverage permits execution; partial coverage narrows privileges to
$\alpha' \subsetneq \alpha$; insufficient coverage fails closed.
Attestation (layer 2) guarantees input integrity but does not determine authority.}
\label{fig:arch}
\end{figure}

\subsection{Execution Flow}

\begin{enumerate}[leftmargin=*, label=(\arabic*)]
  \item Observe $\hat{S}_r(t)$ (partial, bounded) and construct the coverage
        envelope $\mathcal{E}(t) = (\Sp(t), \mathcal{H}(t), \delta(t))$
        via the State Interface Layer.
  \item Verify input integrity via the Attestation Layer (TEE enclave).
        If attestation fails: halt immediately---do not proceed to step 3.
  \item Evaluate $A(t) = F(\mathcal{E}(t), \alpha)$ for the requested
        action class $\alpha$: apply invariants as constructive preconditions.
  \item Reconstruction Gate routes on outcome (see Figure~\ref{fig:arch}):
        \begin{itemize}
          \item $F = \alpha$: execute at full scope;
          \item $F = \alpha' \subsetneq \alpha$: execute with narrowed
                privileges (action class reduced to $\alpha'$);
          \item $F \in \{\text{false}, \Ndefval\}$: halt or re-resolve.
        \end{itemize}
  \item (Optional) Generate justification trace for audit.
\end{enumerate}

\paragraph{Critical property.} No execution step occurs without a fresh
authority construction over the current coverage envelope.
No authority grant from prior steps is carried forward.

\subsection{Failure Handling Model}

\begin{table}[H]
\centering
\caption{RAM reconstruction gate: outcome mapping by coverage condition}
\begin{tabular}{lll}
\toprule
\textbf{Coverage condition} & \textbf{Gate outcome} & \textbf{Behavior} \\
\midrule
Attestation fails (input tampered)        & --- & Halt immediately \\
$F(\mathcal{E}(t),\alpha) = \alpha$        & Full coverage & Execute (full $\alpha$) \\
$F(\mathcal{E}(t),\alpha) = \alpha' \subsetneq \alpha$ & Partial coverage & Execute (narrowed to $\alpha'$) \\
Invariant not satisfiable                 & $\text{false}$ & Halt (definitive) \\
Required component unobservable           & $\Ndefval$ & Halt (insufficient info) \\
\bottomrule
\end{tabular}
\end{table}

\subsection{Where Attestation Still Matters}

Attestation retains an essential role in the hybrid architecture:
\begin{itemize}
  \item Preventing manipulation of the constructor function $F$;
  \item Guaranteeing authenticity of observed state inputs;
  \item Protecting the execution gate from external interference.
\end{itemize}

Attestation cannot, however, decide whether execution is semantically valid.
That decision belongs exclusively to the authority constructor.

\section{Case Study: Autonomous Financial Transfer}
\label{sec:case}

We instantiate the RAM framework on a high-stakes scenario to illustrate the
failure modes of attestation-based systems and the behavior of RAM under the
same conditions.

\subsection{Scenario}

An autonomous agent executes a financial transfer based on:
user risk score, account state, regulatory conditions, and transactional context
(location, device, behavioral pattern).

\paragraph{At $t_0$ (admission).}
Initial state: risk = low, account active, jurisdiction permitted,
behavior consistent.
Decision: transfer of \$10,000 authorized.

\paragraph{At $t_1$ (pre-execution drift).}
Before execution: user IP changes (new geolocation), anomalous behavior
is detected in a parallel session, a regulatory alert appears indicating a
possible restriction.
Real state changes: $\Sr(t_1) \neq \Sr(t_0)$.

\subsection{Attestation-Based System at $t_1$}

\paragraph{Case A: drift detectable.}
If the IP change is included in $\Sp$, a mismatch is detected and execution
halts correctly.

\paragraph{Case B: drift not covered (critical failure).}
Suppose the model does not include multi-session behavior correlation, or the
regulatory signal arrives outside the proof system.
Then $\Sp(t_1) \approx \Sp(t_0)$, attestation passes, execution proceeds---yet
the transfer is invalid under $\Sr(t_1)$.

\begin{center}
\textit{Outcome: correct execution over an incorrect basis.}
\end{center}

\subsection{Oracle-Extended System at $t_1$}

The oracle extends coverage, but if the fraud signal has not yet been included
in the oracle's model or has not propagated, comparison passes and execution
continues.

\begin{center}
\textit{Outcome: perfect validation against an incomplete state model.}
\end{center}

\subsection{RAM-Based System at $t_1$}

The authority constructor $F$ requires: identity consistency, behavior
stability, regulatory compliance, context integrity.
At $t_1$, conflicting signals are observed and the state is partially unknown.

$F(\hat{S}_r(t_1))$ cannot be fully computed: $A(t_1) = \Ndefval$.

Execution halts.

\begin{center}
\textit{Outcome: execution deferred until authority can be established.}
\end{center}

\subsection{Comparative Outcome and Edge Case}

\begin{table}[H]
\centering
\caption{Model comparison under drift: Case B (hidden drift)}
\begin{tabular}{lclp{5.5cm}}
\toprule
\textbf{Model} & \textbf{Executes?} & \textbf{Correct?} & \textbf{Failure mode} \\
\midrule
Attestation (closed) & Yes & No & Executes on stale $\Sp$; $\delta(t)$ undetected \\
Attestation + Oracle & Yes (in many cases) & No & Oracle coverage gap; unmodeled signal undetected \\
RAM                  & No  & Yes & None (defers until authority reconstructible) \\
\bottomrule
\end{tabular}
\end{table}

\paragraph{Edge case: legitimate change.}
If IP change is legitimate, anomaly is noise, and all conditions remain valid,
then: Attestation executes (correctly), RAM reconstructs authority (correctly)
and executes.
RAM can both halt when it should not execute \emph{and} execute when it should.
Attestation can only execute unless it detects a problem.

\paragraph{Core observation.}
The failure does not arise because the system fails to detect drift or because
attestation fails.
It arises because the system \emph{never reconsiders} whether execution is
still justifiable under the current state.

\[
\exists\; \Sr(t_1)\text{ such that execution is invalid, but valid under any }
\Sp(t_1)\text{ available.}
\]

Therefore: no system operating solely on $\Sp$ can prevent this failure.

\section{Experimental Evaluation}
\label{sec:exp}

\subsection{Environment Model}

We simulate a decision system with real state:
\[
\Sr = \{I,\; B,\; R,\; C,\; E\}
\]
where $I$ = identity consistency, $B$ = behavioral patterns,
$R$ = regulatory compliance, $C$ = transactional context,
$E$ = emergent factors (not always observable).

Each model accesses a subset:
\begin{itemize}
  \item Attestation: $\Sp \subsetneq \Sr$;
  \item Oracle: $\Sp \cup \Sp^{\text{ext}} \subsetneq \Sr$;
  \item RAM: attempts to use $\Sr$ with bounded uncertainty.
\end{itemize}

\subsection{Reproducibility}

Simulations were run for $N = 10{,}000$ steps per model with a fixed random
seed (seed~$= 42$) to ensure reproducibility.
Drift events are drawn independently per step according to the probability
distribution in Table~\ref{tab:drift} using a seeded pseudo-random number
generator; no external datasets are required.

\subsection{Drift Injection Model}

We inject four types of drift between $t_0$ and $t_1$:

\begin{table}[H]
\centering
\caption{Drift type distribution in simulation}
\label{tab:drift}
\begin{tabular}{llc}
\toprule
\textbf{Type} & \textbf{Description} & \textbf{Probability} \\
\midrule
Observable  & Changes in $I$, $C$ (directly in $\Sp$)                & 30\% \\
Delayed     & Changes in $R$ (arrive after attestation window)        & 25\% \\
Hidden      & Changes in $E$ (not in any provable model)              & 25\% \\
Ambiguous   & Weak signals in $B$ (insufficient for classification)   & 20\% \\
\bottomrule
\end{tabular}
\end{table}

\subsection{Metrics}

\[
\IER = \frac{|\{\text{exec.}: \Ar = \text{false}\}|}{|\text{executions}|}
\qquad
\SHR = \frac{|\{\text{halts}: \Ar = \text{false}\}|}{|\{\Ar = \text{false}\}|}
\qquad
\OCR = \frac{|\{\text{halts}: \Ar = \text{true}\}|}{|\{\Ar = \text{true}\}|}
\]

$\IER$ (lower is better) measures dangerous invalid executions.
$\SHR$ (higher is better) measures correct safety interventions.
$\OCR$ (lower is better) measures over-conservatism cost.

\subsection{Authority Reconstruction Function (Pseudocode)}

\begin{lstlisting}[caption={RAM authority construction function},
                   label={lst:ram}]
def construct_authority(state):
    # Required components for authority
    required = [
        "identity_consistency",
        "behavior_stability",
        "regulatory_compliance",
        "context_integrity"
    ]

    # Check observability of all required components
    for component in required:
        if component not in state or state[component] is None:
            return UNDEFINED  # insufficient information -> halt

    # Evaluate constructive conditions (not validation)
    if not state["identity_consistency"]:
        return False  # definite refusal

    if not state["behavior_stability"]:
        return UNDEFINED  # uncertainty -> cannot establish authority

    if not state["regulatory_compliance"]:
        return False  # definite refusal

    if not state["context_integrity"]:
        return UNDEFINED  # uncertainty -> cannot establish authority

    return True  # authority established

def execution_loop(system):
    while system.active:
        state = observe_state()            # partial S_r(t)
        if not attest(state):
            halt("attestation_failure")
        authority = construct_authority(state)
        if authority is True:
            execute_step()
        else:
            halt("authority_not_constructible")
\end{lstlisting}

\subsection{Results}

We executed $N=100{,}000$ simulation steps per coverage level with fixed seed~42.

\begin{table}[H]
\centering
\caption{Simulation results at low coverage ($|\Sp|/|\Sr|=0.1$), $N=100{,}000$, seed~42}
\label{tab:results-baseline}
\begin{tabular}{lccc}
\toprule
\textbf{Model} & \textbf{IER} & \textbf{SHR} & \textbf{OCR} \\
\midrule
Attestation        & 0.423 & 0.145 & 0.053 \\
Oracle-extended    & 0.393 & 0.293 & 0.111 \\
RAM                & \textbf{0.000} & \textbf{1.000} & 0.000 \\
\bottomrule
\end{tabular}
\end{table}

\begin{figure}[H]
\centering
\begin{subfigure}[b]{0.48\textwidth}
  \includegraphics[width=\textwidth]{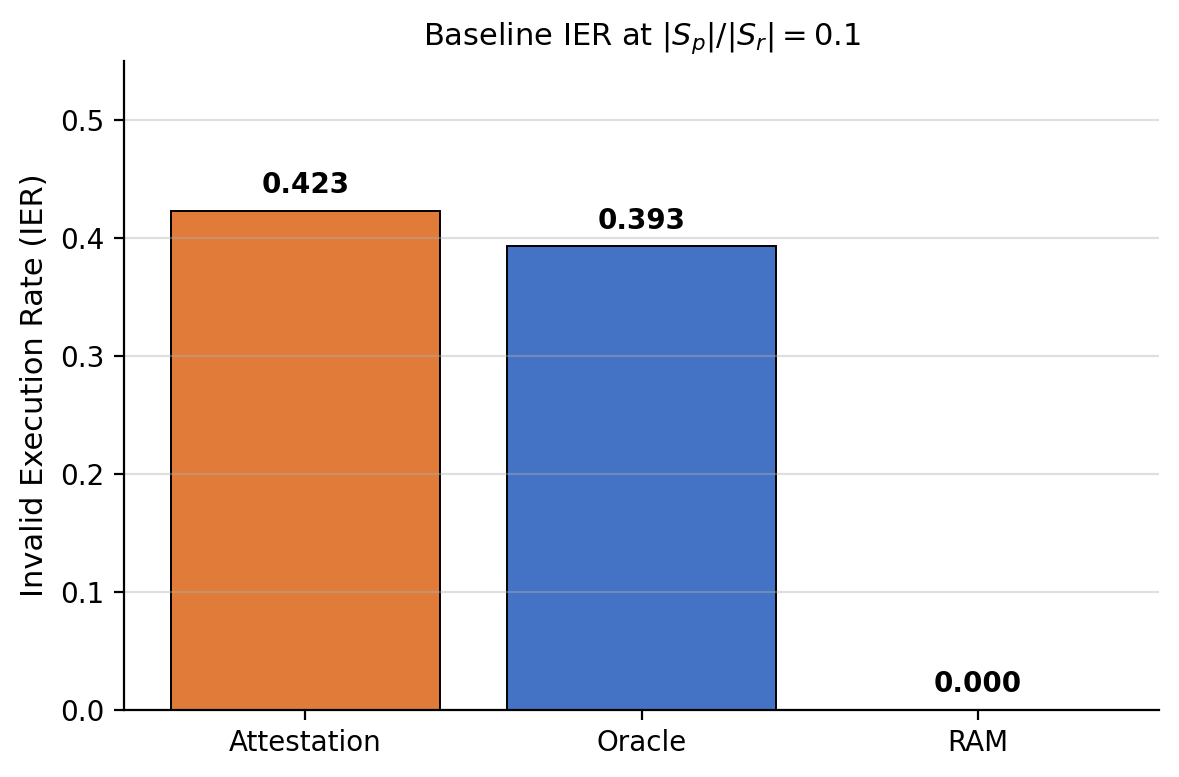}
  \caption{Invalid Execution Rate at baseline coverage\\
           (Attestation~$=$~0.423, Oracle~$=$~0.393, RAM~$=$~0.000)}
  \label{fig:ier-bar}
\end{subfigure}
\hfill
\begin{subfigure}[b]{0.48\textwidth}
  \includegraphics[width=\textwidth]{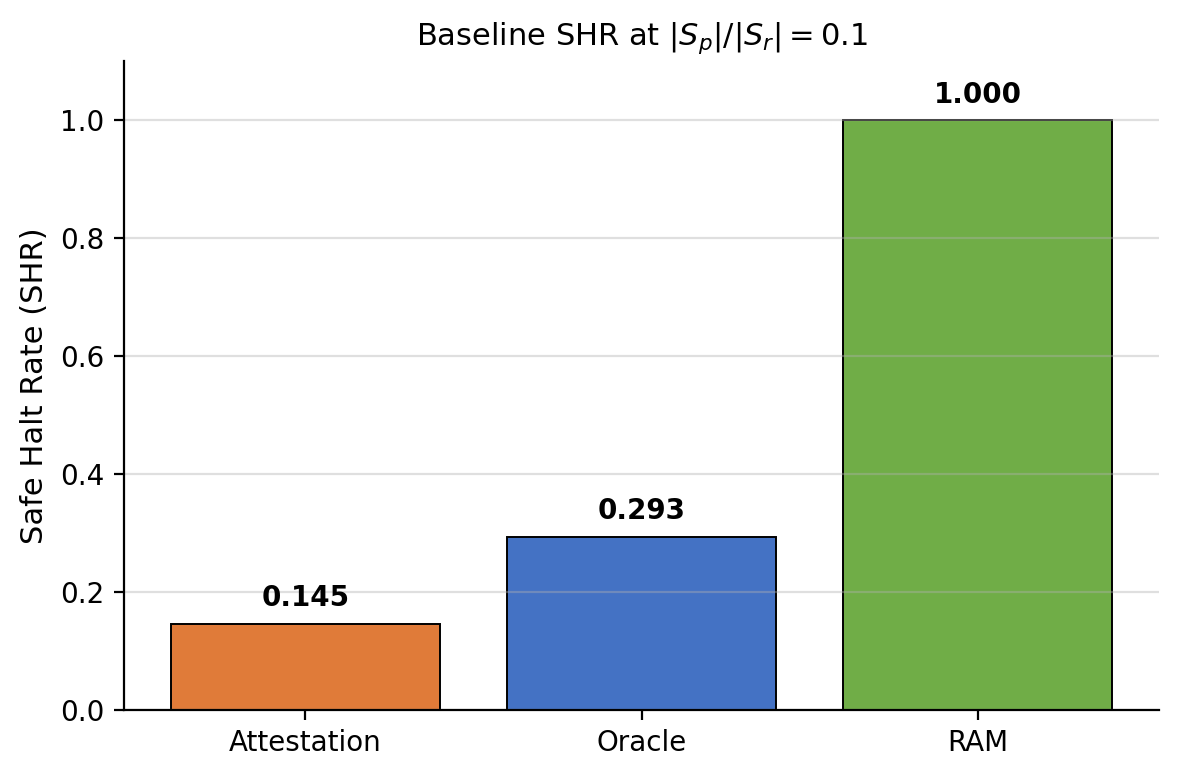}
  \caption{Safe Halt Rate at baseline coverage\\
           (Attestation~$=$~0.145, Oracle~$=$~0.293, RAM~$=$~1.000)}
  \label{fig:shr-bar}
\end{subfigure}
\caption{Baseline comparison at $|\Sp|/|\Sr|=0.1$.
Attestation executes 42.3\% of cases that are actually invalid.
RAM halts all invalid executions with SHR~$=$~1.}
\label{fig:baseline}
\end{figure}

\begin{figure}[H]
\centering
\begin{subfigure}[b]{0.48\textwidth}
  \includegraphics[width=\textwidth]{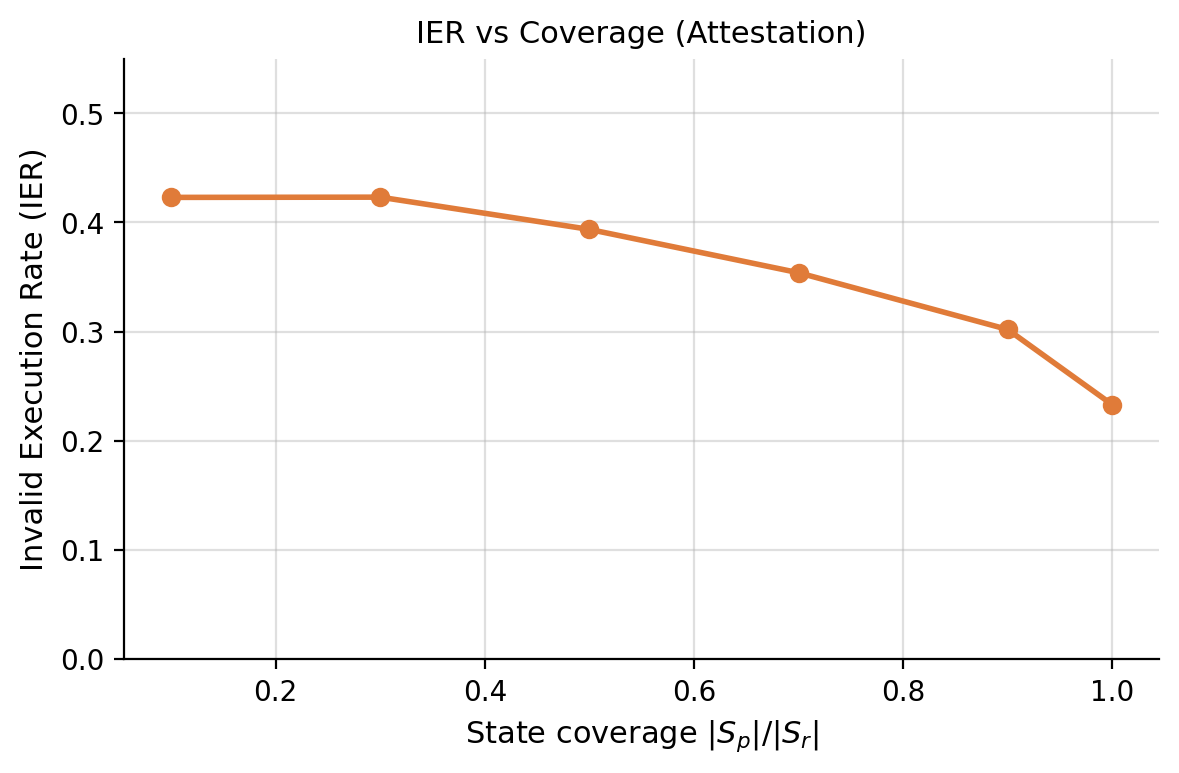}
  \caption{IER vs.\ state coverage $|\Sp|/|\Sr|$
           (attestation model). Error persists even at high coverage
           due to ambiguous/undefined state.}
  \label{fig:ier-coverage}
\end{subfigure}
\hfill
\begin{subfigure}[b]{0.48\textwidth}
  \includegraphics[width=\textwidth]{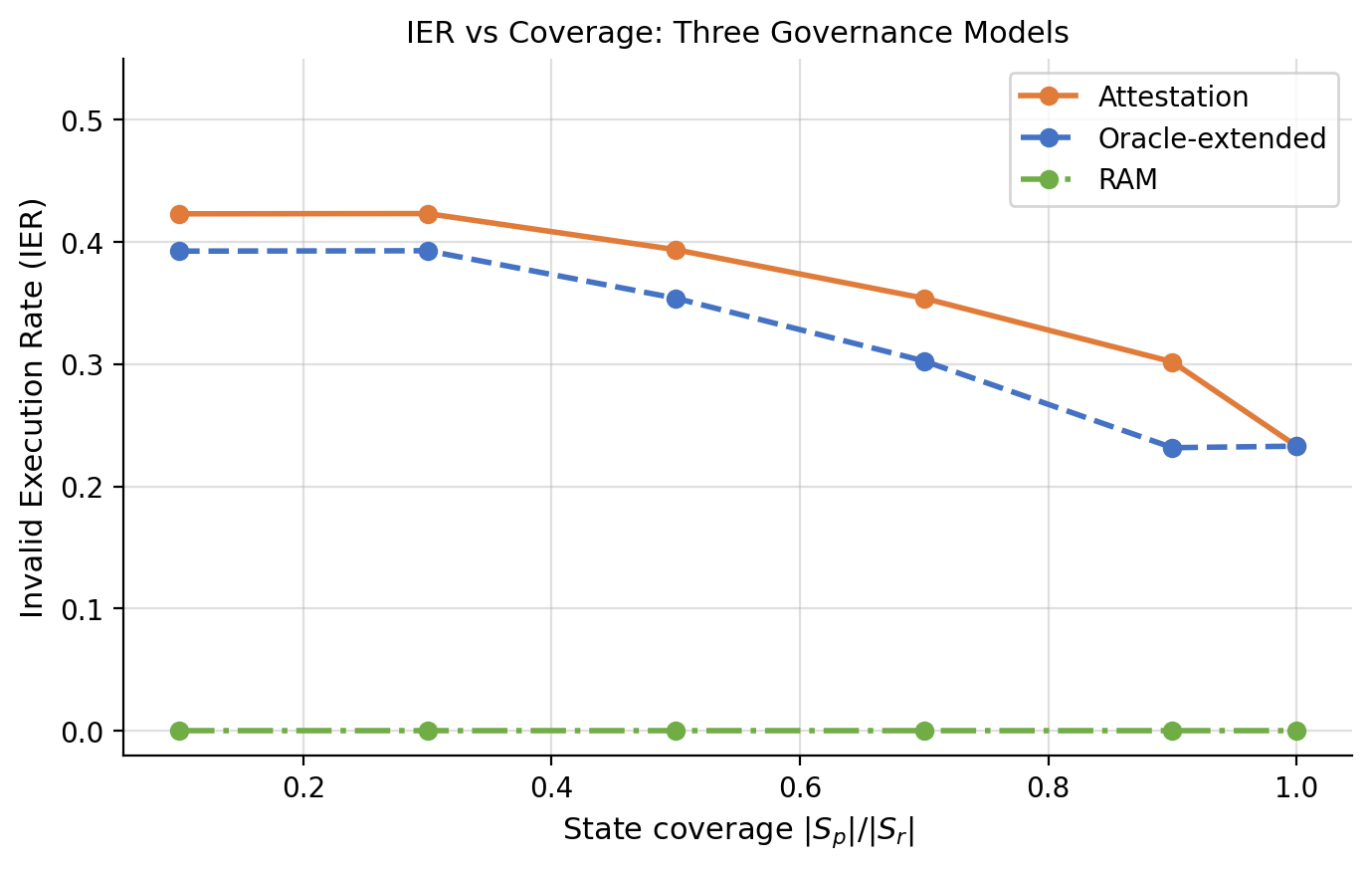}
  \caption{IER vs.\ coverage: Attestation vs.\ Oracle vs.\ RAM.
           RAM maintains IER~$=$~0 across all coverage levels.}
  \label{fig:combined}
\end{subfigure}
\caption{Invalid Execution Rate as a function of state coverage $|\Sp|/|\Sr|$.
Attestation IER remains above 0.23 even at full coverage because undefined
(\texttt{UNDEFINED}) state is not treated as invalid by attestation.
RAM's zero-IER guarantee holds unconditionally.}
\label{fig:coverage}
\end{figure}

\subsection{Key Findings}

\begin{enumerate}[leftmargin=*, label=(\roman*)]
  \item \textbf{Attestation}: IER~$=$~0.423 at low coverage ($|\Sp|/|\Sr|=0.1$),
        declining to 0.233 at full coverage.
        Critically, IER does not reach zero even at $|\Sp|/|\Sr|=1.0$: attestation
        treats undefined state as safe (``not False'' $\neq$ True), while RAM
        correctly halts on any undefined authority component.
        Safe Halt Rate~$=$~0.145: the system misses 85\% of invalid executions
        it cannot detect.

  \item \textbf{Oracle-extended}: IER~$=$~0.393 at baseline, declining to 0.233 at
        full coverage.
        The oracle extension improves coverage quantitatively but exhibits the same
        undefined-state failure: both attestation and oracle converge to
        IER~$=$~0.233 at $|\Sp|/|\Sr|=1.0$, demonstrating that the structural
        gap is not coverage-dependent alone.

  \item \textbf{RAM}: IER~$=$~0.000 at \emph{every} coverage level, SHR~$=$~1.000.
        RAM correctly treats undefined authority components as insufficient for
        execution, halting rather than proceeding under uncertainty.
        In this simulation, OCR~$=$~0.000 because the authority function $F$ is
        precisely scoped to the components that determine real authority.
        In deployments where $F$ applies broader safety criteria, OCR~$>$~0 and
        decreases as observability improves (see Section~\ref{sec:ocr}).
\end{enumerate}

\paragraph{Core experimental insight.}

Attestation-based systems exhibit a two-part failure: a coverage-dependent
component (IER decreasing with $|\Sp|/|\Sr|$) and a semantic component that
persists at full coverage (IER~$=$~0.233 even when all state is observable).
The semantic failure arises because attestation applies a weaker validity
criterion (\emph{not provably false}) rather than the constructive criterion
required for authority (\emph{provably true}).
RAM eliminates both failure modes by construction.

\[
\IER_{\text{attestation}} \propto \left(1 - \frac{|\Sp|}{|\Sr|}\right) + \varepsilon_{\text{semantic}}
\]

where $\varepsilon_{\text{semantic}} > 0$ whenever undefined state can be
treated as benign by the attestation decision function.

\paragraph{Critical result.}
There exists a non-empty class of scenarios (hidden drift, ambiguous/undefined
state) where attestation and oracle-extended systems execute incorrectly while
RAM halts correctly.
This class is observable in simulation and predicted by
Theorem~\ref{thm:attestation-insufficiency}.

\subsection{The Security--Execution Trade-off (OCR)}
\label{sec:ocr}

\begin{figure}[H]
\centering
\includegraphics[width=\textwidth]{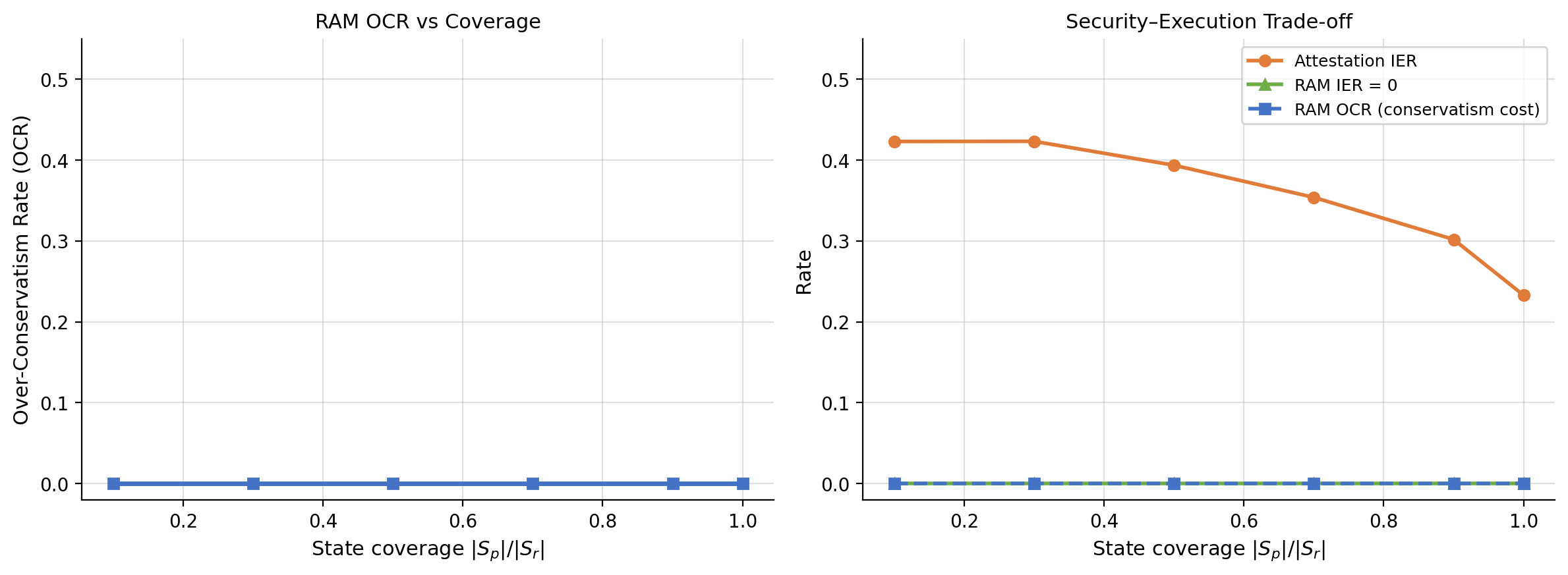}
\caption{Left: RAM Over-Conservatism Rate (OCR) as a function of coverage.
In the idealized simulation (seed~42, $N=100{,}000$), OCR~$=$~0 because the
authority function $F$ is precisely scoped to execution-critical components.
In practice, OCR~$> 0$ whenever $F$ applies broader safety criteria than
strictly necessary for the specific action class.
Right: Complete trade-off surface showing attestation IER (orange, decreasing
with coverage but never zero) and RAM IER (green, constant at 0). As coverage
improves, RAM's conservatism cost decreases toward zero while its zero-IER
guarantee is maintained unconditionally.}
\label{fig:ocr}
\end{figure}

The OCR trade-off is structurally real even when not observed in this
simulation.
OCR arises whenever RAM's authority function $F$ requires confirming state
components that are not strictly necessary for the specific action class.
This is a \emph{deployment design decision}: a more conservative $F$ (broader
scope) incurs higher OCR; a more precise $F$ (narrower scope) reduces OCR at
the cost of potentially missing edge cases.

In practice, OCR decreases as observability improves---the same observability
investment that reduces attestation's IER also allows RAM to reconstruct
authority with higher confidence, reducing unnecessary halts.
The critical asymmetry: RAM's IER~$= 0$ is an unconditional guarantee
independent of coverage, while attestation's IER~$= 0$ requires asymptotically
complete coverage and correct handling of undefined state.

\paragraph{OCR as investment signal.}
At any given deployment, RAM's OCR measures the gap between current
observability and the minimum required for zero-cost authority reconstruction.
A persistently high OCR is a signal to invest in observability engineering,
not to override the halt.

\section{Implications for System Design}
\label{sec:design}

\paragraph{1. Authority must be stateless at runtime.}
Systems must not depend on prior decisions, sealed contexts, or persistent
authority tokens.
Any mechanism that ``carries'' authority forward in time introduces the risk
of executing on stale grounds.

\paragraph{2. Validation is insufficient without construction.}
Validating policies, constraints, and envelopes is not enough.
The system must be able to \emph{construct} authority from the current state,
not merely verify consistency.

\paragraph{3. Observability must be explicitly modeled.}
Every deployment must define which state components are:
observable, non-observable, and inferred.
If a critical component is not observable, authority cannot be constructed:
halt, not proceed.

\paragraph{4. Default behavior must be conservative.}
Classical model: ``if no error detected, continue.''
RAM model: ``if authority cannot be established, halt.''
For high-stakes autonomous systems, the conservative default is the only
safe option consistent with Theorem~\ref{thm:necessity}.

\paragraph{5. Attestation becomes a supporting layer, not the source of truth.}
Attestation ensures that what executes was not manipulated.
RAM ensures that what executes is still valid.
Both are necessary; neither alone is sufficient.

\paragraph{6. Drift is not a signal to monitor---it is a failure to construct.}
Systems do not need to detect, classify, or respond to drift.
They simply fail to reconstruct authority.
This simplifies system logic and eliminates entire classes of race conditions
and incomplete drift detectors.

\section{Discussion}
\label{sec:discussion}

\subsection{Convergence with Attestation Under High Coverage}

Under near-complete observability ($|\Sp|/|\Sr| \to 1$), RAM and attestation
models converge behaviorally: both exhibit low IER and low OCR.
In environments with comprehensive state instrumentation, the practical gap
between RAM and oracle-extended attestation narrows.

However, the asymmetry is structural:
RAM's $\IER = 0$ guarantee holds by construction regardless of coverage level.
Attestation's guarantee is coverage-dependent and collapses whenever an
unmodeled state dimension becomes execution-critical.
For any system where hidden or delayed drift is possible---which is every
real-world deployment---RAM's structural guarantee is the only one that holds
unconditionally.

\subsection{Limitations}

\paragraph{Computational overhead.}
Authority reconstruction at every execution step introduces per-step overhead
proportional to the complexity of $F$.
For high-frequency decision systems, this may be significant.
Partial caching is viable when state components change slowly, but cache
invalidation must be conservative (invalidate on any observable change).

\paragraph{Persistent halting under insufficient observability.}
If $F$ requires state components that are permanently unobservable in a given
deployment, the system will halt indefinitely.
This is a feature, not a bug: it surfaces \emph{observability debt} explicitly.
A system that cannot reconstruct authority for a given action class is
communicating a design fact---the instrumentation is inadequate for the
authority model.
The correct response is to improve observability or relax authority requirements,
not to override the halt.
Overriding the halt restores the attestation failure mode.

\paragraph{Scalability in multi-agent settings.}
In multi-agent systems, $\hat{S}_r(t)$ must aggregate state across agents,
introducing coordination overhead and potential inconsistency windows.
RAM does not solve distributed state aggregation; it provides a principled
response to its structural limitations.

\paragraph{Synthetic experimental basis.}
The simulation in Section~\ref{sec:exp} models idealized drift distributions.
Real systems exhibit drift patterns, coverage profiles, and authority function
complexities that differ from our parametric model.
Empirical validation on production agent workloads remains an open task.

\subsection{Future Work}

\begin{enumerate}[leftmargin=*, label=(\roman*)]
  \item \textbf{RAM + ACP integration}: end-to-end governance pipelines where
        ACP governs admission and RAM governs runtime, with formal composition
        guarantees.
        This closes the series: P1 provides the admission gate; P5 provides the
        execution gate; their integration yields a complete authority lifecycle.
  \item \textbf{OCR--IER Pareto frontier}: formal characterization of the
        trade-off surface as a function of observability investment, enabling
        cost-optimal deployment decisions.
  \item \textbf{Empirical validation}: evaluation on real agent workloads
        including financial transaction systems, autonomous vehicle decision
        pipelines, and LLM agent tool execution logs.
  \item \textbf{Adversarial state injection}: RAM behavior under adversarial
        manipulation of $\hat{S}_r(t)$, where an attacker injects false
        observations to induce incorrect authority construction.
  \item \textbf{RAM + learned state inference}: augmenting $F$ with a learned
        estimator of unobservable components, allowing RAM to reason under
        probabilistic state completeness while preserving the conservative
        default.
  \item \textbf{Multi-agent RAM coordination}: distributed authority
        construction where $A(t)$ requires consensus across agents with
        overlapping but distinct observations of $\Sr(t)$.
\end{enumerate}

\section{Relation to Governance Series}
\label{sec:series}

RAM completes the Agent Governance Series by providing the \emph{operational
closure}: the mechanism that answers the runtime execution question under the
theoretical constraints established by Papers 0--4.

\paragraph{Paper 0 (Atomic Boundaries)~\cite{fernandez2026a}.}
P0 establishes that decisions must be atomic at the enforcement boundary.
RAM operates at the post-admission stage: once an atomic decision is made,
RAM determines whether that decision remains valid at execution time.

\paragraph{Paper 1 (ACP)~\cite{fernandez2026b}.}
ACP implements enforcement via admission control: authority is granted at
$t_0$ based on the state at admission.
RAM is the runtime dual of ACP: authority is re-derived at each $t$ based
on the state at that moment.
ACP governs \emph{whether to admit}; RAM governs \emph{whether to execute}.

\paragraph{Paper 2 (IML)~\cite{fernandez2026c}.}
IML proves that for any governance architecture with finite memory and local
observation, full state coverage is unachievable.
RAM operationalizes this result: given that $\Sp \subsetneq \Sr$ is a
permanent condition (IML Theorem), execution authority must be designed to
fail gracefully under incomplete observation.
RAM is the constructive response to IML's impossibility result.

\paragraph{Paper 3/4 (Governance Structure)~\cite{fernandez2026govstr}.}
Paper~3/4 establishes two complementary results. First, strategy-proof allocation
mechanisms face irreducible distributive limits under Sybil amplification; RAM's
conservative halting can be combined with its fair allocation framework to ensure
halts are distributed equitably across agents. Second, it proves that the
four-layer governance architecture (L0--L3) is irreducible under finite
observability. RAM operates as the execution authority layer within L2 (behavioral
control), interacting with L1 (admission/ACP) and L3 (compositional invariants),
and completes the execution semantics of the four-layer architecture.

\paragraph{Series summary.}
\begin{center}
\small
\begin{tabular}{cll}
\toprule
\textbf{Paper} & \textbf{Question addressed} & \textbf{Formal contribution} \\
\midrule
P0 & When is a decision atomic? & Atomic boundary theorem \\
P1 & How to enforce the boundary? & ACP admission protocol \\
P2 & What can be observed? & Observability impossibility result \\
P3/4 & Is governance structure irreducible? & Fair allocation + compositional irreducibility \\
P5 & Given all of the above, when to execute? & RAM + attestation necessity \\
\bottomrule
\end{tabular}
\end{center}

\section{Conclusion}
\label{sec:conclusion}

We have established that attestation-based governance is structurally limited
by the gap between provable state $\Sp$ and real state $\Sr$.
Attestation can guarantee computational integrity; it cannot guarantee
execution validity when the invalidating conditions lie outside the attested
model.

The Reconstructive Authority Model (RAM) resolves this by repositioning
execution authority as a continuously derived property:
$A(t) = F(\Sr(t))$.
RAM eliminates invalid execution by construction, at the cost of conservative
halting under uncertainty---a trade-off we argue is the only defensible default
for high-stakes autonomous systems.

The key theorems establish:
(i)~attestation correctness does not imply execution validity
(Theorem~\ref{thm:attestation-insufficiency});
(ii)~reconstruction from $\Sr(t)$ at runtime is a \emph{necessary} condition
for execution validity guarantees (Theorem~\ref{thm:necessity}).
These are structural results, not implementation-specific observations.

Experimentally, RAM achieves $\IER = 0$ across all state coverage levels,
while attestation-based systems exhibit $\IER \propto (1 - |\Sp|/|\Sr|)$.
The class of failures eliminated by RAM---those arising from hidden drift
and emergent conditions outside the provable state---cannot be addressed by
stronger attestation or more comprehensive oracles alone.

\paragraph{Final statement.}
\begin{center}
\textit{Attestation guarantees that we measure correctly.}\\[0.2em]
\textit{RAM guarantees that what we measure is enough to act.}
\end{center}

\paragraph{Future work.}
Open directions are detailed in Section~\ref{sec:discussion}:
OCR--IER Pareto frontier characterization, RAM+ACP end-to-end integration,
empirical validation on production workloads, adversarial state injection
analysis, RAM with learned state inference, and multi-agent RAM coordination.

\appendix
\section{Formal Proof of Theorem~\ref{thm:attestation-insufficiency}}
\label{app:proof}

We provide the complete constructive proof of Theorem~1, which was presented
as a proof sketch in the main body.

\paragraph{Setup.}
Let $\mathcal{S}$ denote the full state space.
Let $\Sr(t) \in \mathcal{S}$ be the real execution-relevant state at time $t$.
Let $\Sp(t) \subseteq \Sr(t)$ be the provable (attested) state.
Define the state gap $\delta(t) = \Sr(t) \setminus \Sp(t)$, which is non-empty
by Assumption~1 (partial observability).

Let $G: 2^{\mathcal{S}} \to \{\text{true}, \text{false}\}$ be a deterministic
authority function evaluated over $\Sp(t)$ (the attestation decision function).
Let $F: 2^{\mathcal{S}} \to \{\text{true}, \text{false}\}$ be the real authority
function that determines validity with respect to $\Sr(t)$.

An \emph{invalid execution} occurs when $G(\Sp(t)) = \text{true}$ but
$F(\Sr(t)) = \text{false}$.

\paragraph{Lemma~A.1 (Gap existence).}
\textit{For any governance architecture with finite memory and local observation,
$\delta(t) \neq \emptyset$ in general.}

\begin{proof}
By the epistemological impossibility result of Paper~2 (IML,~\cite{fernandez2026c}),
no finite-memory local observer can fully reconstruct $\Sr(t)$ in a system with
emergent, delayed, or correlated state components.
Formally: for any collection of observations $\{o_i\}_{i=1}^k$ with $k < |\Sr(t)|$,
there exists at least one component $s^* \in \Sr(t)$ not captured by any $o_i$.
Therefore $\Sp(t) = \bigcup_i \text{obs}(o_i) \subsetneq \Sr(t)$,
i.e., $\delta(t) \neq \emptyset$.
\end{proof}

\paragraph{Lemma~A.2 (Execution-critical gap).}
\textit{There exists a state $\Sr^*(t)$ and a component $\delta^*(t) \in \delta(t)$
such that:}
\begin{enumerate}[leftmargin=*, label=(\alph*)]
  \item $G(\Sp(t)) = \text{true}$ \quad (attestation approves);
  \item $F(\Sr^*(t)) = \text{false}$ \quad (real authority is false);
  \item $F(\Sp(t)) = \text{true}$ \quad (the invalidating condition is entirely
        within $\delta^*(t)$, not in $\Sp(t)$).
\end{enumerate}

\begin{proof}
Construct $\Sr^*(t)$ as follows.
Let $\Sp(t)$ be fixed such that $G(\Sp(t)) = \text{true}$ (such a state exists
whenever the system would admit execution under normal conditions).
Since $\delta(t) \neq \emptyset$ by Lemma~A.1, pick any $s^* \in \delta(t)$.
Define $\Sr^*(t) = \Sp(t) \cup \{s^*\}$ and set $s^*$ to a value such that
$F(\Sr^*(t)) = \text{false}$, i.e., $s^*$ represents an execution-critical
condition (e.g., regulatory restriction, fraud signal, invariant violation)
that renders execution invalid.

Since $s^* \notin \Sp(t)$, we have $F(\Sp(t)) = \text{true}$ (conditions~(a)
and~(c) hold by construction).
Since $s^*$ invalidates authority under the full real state, condition~(b) holds.
Such a value for $s^*$ exists whenever $F$ is sensitive to at least one
state component not present in $\Sp(t)$---a necessary condition for any
non-trivial authority function over a richer state space.
\end{proof}

\paragraph{Main proof.}

\begin{proof}[Proof of Theorem~\ref{thm:attestation-insufficiency}]
We proceed by constructive counterexample.
By Lemma~A.1, $\delta(t) \neq \emptyset$.
By Lemma~A.2, construct $(\Sp(t), \Sr^*(t), \delta^*(t))$ satisfying
conditions~(a)--(c).

Under this construction:
\begin{itemize}
  \item The attestation system evaluates $G(\Sp(t)) = \text{true}$ and permits
        execution.
  \item The real authority function evaluates $F(\Sr^*(t)) = \text{false}$:
        execution is invalid.
  \item Since $\delta^*(t) \notin \Sp(t)$, no attestation mechanism operating
        solely over $\Sp(t)$---regardless of its cryptographic strength or
        completeness within $\Sp(t)$---can detect the invalidating condition.
  \item Therefore, execution proceeds (under attestation) while being invalid
        (under real state).
\end{itemize}

This establishes that attestation correctness over $\Sp(t)$ does not imply
execution validity over $\Sr(t)$.

The only way to close this gap is to make $F$ computable from $\Sr(t)$ directly.
Since $\Sr(t)$ is not fully observable ($\delta(t) \neq \emptyset$), any function
that guarantees $F(\Sr(t)) = \text{true}$ before executing must return
$\Ndefval$ whenever $\hat{S}_r(t) \subsetneq \Sr(t)$ is insufficient to
determine the output of $F$.
This is precisely the RAM execution gate.
\end{proof}

\paragraph{Remark on the constructive case with $\delta(t)$.}
The construction in Lemma~A.2 is explicit: $\delta^*(t)$ is a single
execution-critical state component that is permanently excluded from $\Sp(t)$.
In real systems, $\delta(t)$ may contain:
\begin{itemize}
  \item \textit{Emergent state}: conditions that arise after $t_0$ and are
        not instrumentable before $t_1$ (e.g., fraud signals, regulatory changes);
  \item \textit{Delayed state}: components whose values are not available at
        attestation time due to propagation latency;
  \item \textit{Correlated state}: variables whose relevance becomes apparent
        only when combined with other unobservable components.
\end{itemize}
In all cases, the structure of the proof is identical: $\delta^*(t)$ is
execution-critical, unobservable, and therefore undetectable by any
attestation-based system.

\bibliographystyle{unsrtnat}
\bibliography{references}

\end{document}